\documentclass{ecai} 

\usepackage{xcolor}
\usepackage{csquotes}
\usepackage{algorithm2e}
\usepackage{latexsym}
\usepackage{amssymb}
\usepackage{amsmath}
\usepackage{mathtools}
\usepackage{amsthm}
\usepackage{booktabs}
\usepackage{enumitem}
\usepackage{graphicx}
\usepackage{color}
\usepackage{subcaption}
\usepackage{tikz}
\usepackage{multirow}
\newcounter{cntr}

\newtheorem{theorem}{Theorem}
\newtheorem{lemma}[theorem]{Lemma}

\newtheorem{definition}{Definition}

\newcommand{\BibTeX}{B\kern-.05em{\sc i\kern-.025em b}\kern-.08em\TeX}
\newcommand{\connectionpoint}[1]{c_\odot{#1}}
\newcommand{\rh}[3]{\text{rh}_{#1, #2}\text{(}#3\text{)}}
\newcommand{\rrh}[2]{\text{rrh}_{#1}\text{(}#2\text{)}}
\newcommand{\mrrh}[1]{\text{rrh(}#1\text{)}}

\begin{document}


\begin{frontmatter}


\paperid{2613} 


\title{Reaching New Heights in Multi-Agent Collective Construction}


\author[A]{\fnms{Martin}~\snm{Rameš}\orcid{0009-0000-3301-6269}\footnote{Author. Email: ramesmar@fit.cvut.cz}}
\author[A]{\fnms{Pavel}~\snm{Surynek}\orcid{0000-0001-7200-0542}\thanks{Corresponding Author. Email: pavel.surynek@fit.cvut.cz}}

\address[A]{Faculty of Information Technology, Czech Technical University, Thákurova 9, 160 00 Prague 6, Czechia}


\begin{abstract}
We propose a new approach for multi-agent collective construction, based on the idea of reversible ramps. Our ReRamp algorithm utilizes reversible side-ramps to generate construction plans for ramped block structures higher and larger than was previously possible using state-of-the-art planning algorithms, given the same building area. We compare the ReRamp algorithm to similar state-of-the-art algorithms on a set of benchmark instances, where we demonstrate its superior computational speed. We also establish in our experiments that the ReRamp algorithm is capable of generating plans for a single-story house, an important milestone on the road to real-world multi-agent construction applications.
\end{abstract}

\end{frontmatter}


\section{Introduction}
Nature inspires progress in many fields of science and multi-agent systems are no different. Termite colonies are the original inspiration for the TERMES robots, a Harvard University project leading to one of the more studied formalizations of multi-agent collective construction \cite{petersen2011termes, koenigexact}.

The multi-agent collective construction (MACC) tasks a group of cooperative agents with building a given block structure. The agents are roughly the size of the blocks, having six main high-level actions. They can \textit{enter} the building area at the border, \textit{move} to a neighbor position, \textit{deliver} their block or \textit{pick up} a block, \textit{move} back to the border, and \textit{leave}. An agent can carry at most one block at a time. All construction happens within a downward gravity field. An agent can move at most one block up or down when moving to the neighbor cell. The delivered or picked-up block is required to be at the same height as the agent at the end/start of the action, respectively. An example of a TERMES construction site and its MACC discretization is shown in figure \ref{fig:TERMES2MACC}.

To get to higher positions, the agents must build ramps. On simple ramps, the agent moves one block up with each edge on the ramp path. This naturally leads to searching for the longest path, when deconstructing tall columns. On a partial grid graph, when some of the columns of the target structure are already placed, the problem of the longest path is hard. In this paper, we aim to show that by making a reasonable concession on the structure's maximum height and dedicating some of the ramp nodes to ramp reversal, we can switch to building ramps with a tree footprint, instead of a path footprint. In essence, we aim to change the hard problem of finding the longest path for a simple ramp to the problem of finding a spanning tree for a compound ramp. By reversing the side ramps like the agent does between figures \ref{fig:forward-side-ramp} and \ref{fig:backward-side-ramp}, we may continue on the central path from the top of the reversed ramp, even when it is not part of the path footprint.

\begin{figure}
    \centering
    \begin{subfigure}[b]{\columnwidth}
        \centering
        \includegraphics[scale=0.25]{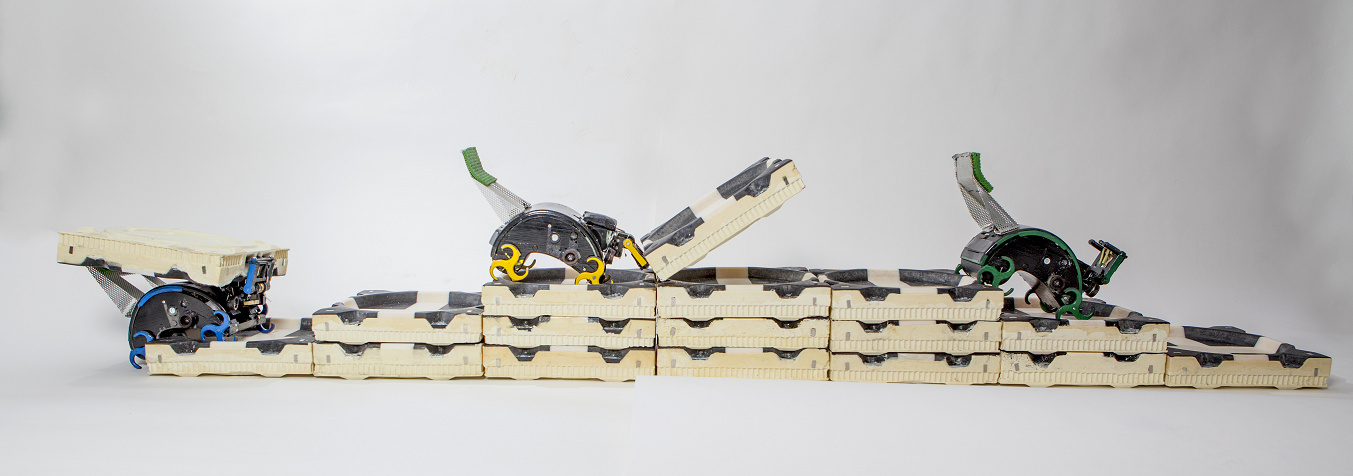}
        \caption{"Termes robot 01" by Forgemind ArchiMedia is licensed under CC BY 2.0 \cite{ImageTERMES}.}
    \end{subfigure}
    \begin{subfigure}[b]{\columnwidth}
        \centering
        \vspace{2mm}
        \includegraphics[scale=1.25]{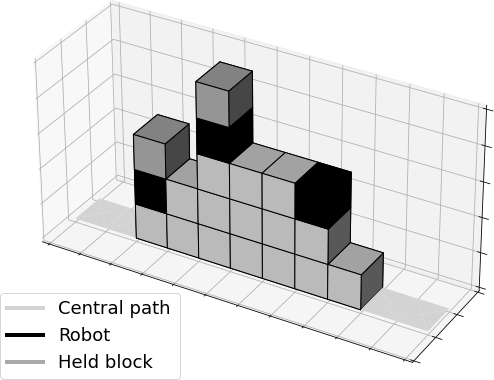}
        \caption{MACC discretization of the construction area.}
    \end{subfigure}
    \vspace{10pt}
    \caption{The TERMES robots and their Multi-Agent Collective Construction discretization. The left agent is carrying a block, the middle agent is placing a block and the right agent is moving without a block.}
    \vspace{15pt}
    \label{fig:TERMES2MACC}
\end{figure}


\section{Related work}
Collaborative construction is a broad research field, encompassing a wide variety of robotic agents with a diverse set of movement and building strategies. For instance, there are robots, which themselves act as building blocks \cite{blockSelfAssembly, SphereRobots}. Some robotic systems use UAVs to build the target structure \cite{flying3dPrinter, flyingBricks, flyingBeamsAndColumns}. Some use them in combination with ground agents \cite{flyingGroundCooperation}. Some proposed ground agents take the form of an inchworm and build lightweight lattice structures \cite{nasaInchword}.

We choose the TERMES system \cite{petersen2011termes} and its associated MACC problem. The main design principle of the TERMES robots is simplicity. The robots can move forward, turn left 90°, turn right 90°, pick up a block just ahead, and deliver a block just ahead \cite{petersen2011termes}. The robots use foam 21.5x21.5x4.5 cm blocks as bricks, with indents and magnets, ensuring stable block placement \cite{petersen2011termes}.

The original TERMES publication provided three rudimentary algorithms for the robots, evaluated in a separate publication \cite{petersen2011termes, TERMESsimpleAlgorithm}. The first algorithm searches for a path connecting all structure nodes without ramps. The second algorithm works with branching and merging of the main path, still without ramps. The third algorithm has a 2-lane simple staircase ramp, deconstructed once the structure is finished.
Among the notable algorithms, that do not use ramps, is the compiler for scalable construction -- capable of solving structures with the building area containing thousands of block columns \cite{deng2019compiler}.
The Tree-Based algorithm for construction robots \cite{TACR} is a polynomial algorithm, that performs dynamic programming on a spanning tree in the inner loop and searches for a good spanning tree in the outer loop. The initial algorithm is single-agent, but there is also a multi-agent variant \cite{multiagentTACR}. At the start, the Tree-Based algorithm generates the workspace matrix, which expands the building area to each side by the structure such that all columns within the building area are at least $h$ blocks from the edge of the building area ($h$ is the height of the column) \cite{TACR}.

There are also three solvers (two in the first paper), which exactly optimize the MACC problem and provide optimal solutions \cite{koenigexact, icaart24}. The first two solve MACC using mixed-integer linear programming (MILP) and constraint programming, respectively; the last one is the generalization of the first, providing optimal solutions even when the agents do not have a constant duration of actions \cite{koenigexact, icaart24}. That being said, the solutions of the first and third models may be slightly different because of the stricter collision avoidance constraints the third model uses \cite{icaart24}.

The last introduced model \cite{decomposition} uses the decomposition before tasking a MILP solver to optimally solve the substructures. The decomposition leads to faster run times at the cost of getting only a suboptimal solution.

For later analysis, we have chosen the big O notation by \cite{Knuth1976BigOA}. It defines the following sets (relevant to this paper) for functions $f$ and $g$:
\begin{itemize}
    \item $O(f(n))$ denotes the set of all $g(n)$ such that there exist positive constants $C$ and $n_0$ with $|g(n)| \leq C f(n), \forall n \geq n_0$
    \item $\Omega(f(n))$ denotes the set of all $g(n)$ such that there exist positive constants $C$ and $n_0$ with $g(n) \geq C f(n)$
\end{itemize}

\subsection{Multi-Agent Collective Construction}

Multi-Agent Collective Construction problem is defined by \cite{koenigexact} as follows: Let there be a three-dimensional grid. Let $X$, $Y$, and $Z$ be the size of the grid in the two horizontal axes and the vertical axis, respectively. Let $\widehat{i}$ be shorthand for $\{0, \dots, i-1\}, i \in \mathbb{Z}_+$. Let $\mathcal{C} = \widehat{X} \times \widehat{Y} \times \widehat{Z}$ be the set of all positions within the grid. Let $\mathcal{P}$ be the top-down projection of $\mathcal{C}$ into the first two dimensions. Let $z_{t, x, y}$ be the height of the block structure at timestep $t$ at position $(x, y) \in \mathcal{P}$. Let $\mathcal{T} = \widehat{T}$ be the planning horizon of $T$ timesteps. Let $\mathcal{B} = \{(x, 0, 0): x \in \mathcal{X}\} \cup \{(x, Y-1, 0): x \in \mathcal{X}\} \cup \{(0, y, 0): y \in \mathcal{Y}\} \cup \{(X-1, y, 0): y \in \mathcal{Y}\}$ be the set of border cells at the perimeter of the building area. Let $\mathcal{N}_{(x, y)} = \{(x - 1, y), (x + 1, y), (x, y - 1), (x, y + 1)\} \cap \mathcal{P}$ be the set of neighbor positions of $(x, y)$.

At each timestep, a grid cell can be either unoccupied, occupied by an agent, or occupied by a block. When a grid cell at $(x, y, z) \in \mathcal{C}$ is occupied, all cells below it (\{$(x, y, z'): 0 \leq z' < z$\}) must be occupied by blocks. The grid starts with all cells unoccupied. The task is to build a given target block structure using agent actions. All agents must leave the grid by the end of the planning horizon. Border cells can be occupied only by agents. At most one object can occupy a grid cell at any given timestep (forbidding agent vertex collisions). An agent can carry at most one block. While the agent carries the block, it is considered to be a part of the agent.

An agent has the following actions:
\begin{itemize}
    \item \textit{enter} through a border cell (agent moves from outside the grid to a border cell while optionally carrying a block)
    \item \textit{leave} through a border cell (agent moves from a border cell to outside the grid while optionally carrying a block)
    \item \textit{deliver} a block at neighbor position $(x', y') \in \mathcal{N}_{(x, y)}$, while the agent stands at $(x, y)$, given $z_{t, x, y} = z_{t, x', y'}$ at the start of the action and the agent holds a block at the start of the action
    \item \textit{pick up} a block at neighbor position $(x', y') \in \mathcal{N}_{(x, y)}$, while the agent stands at $(x, y)$, given $z_{t, x, y} = z_{t, x', y'}$ at the end of the action and the agent does not hold a block at the start of the action
    \item \textit{move} from a cell $(x, y)$ to a neighbor cell $(x', y') \in \mathcal{N}_{(x, y)}$,
    given $|z_{t, x, y} - z_{t, x', y'}| \leq 1$ at the start of the action
    \item \textit{wait} at the same position
\end{itemize}

Let $\mathcal{A} = \{\text{enter}, \text{leave}, \text{deliver}, \text{pick\_up}, \text{move}, \text{wait}\}$ be the set of all action types. For simplicity, unless otherwise stated, the duration of each action is assumed to be one timestep.

For agent edge collision prevention, we choose a more restrictive constraint by \cite{icaart24}, who call the area between $(x, y)$ of action start and $(x', y')$ of action end (inclusive on both ends) the \enquote{exclusion zone}. They propose a constraint that at any given timestep the intersection between any two exclusion zones of currently executing actions must be empty.

The original edge collision prevention constraint by \cite{koenigexact} forbids just position exchange of the agents in one timestep (since the agents would have to go through each other). The original constraint, however, does not prevent collisions in cases like one agent going forward, while the second is moving from its side to its start position.


\section{MACC analysis}

Before defining our proposed algorithm, we perform an analysis of the MACC problem. In particular, we focus on the decision problem variant of MACC, which determines, if the given target structure is \textit{buildable} (there exists a finite sequence of actions $a \in \mathcal{A}$, which meets MACC constraints and ends with the finished target structure and no agents on the grid).

Let grid graph $G = (\mathcal{P}, E)$ be a undirected graph, where $E = \{\{(x, y), (x', y')\}:(x, y) \in \mathcal{P} \wedge (x', y') \in \mathcal{N}_{(x, y)}\}$.

Let $\mathcal{B}_\mathcal{P} = \{(x, y): (x, y, z) \in \mathcal{B}\}$ be a projection of $\mathcal{B}$ to $\mathcal{P}$. Let $\mathcal{P}_\boxdot = \mathcal{P} \setminus \mathcal{B}_\mathcal{P}$. Let $G_{t, 0}$ be an induced subgraph of $G$ on nodes $(x, y) \in \mathcal{P}_\boxdot$, where $z_{t, x, y} = 0$. A \textit{ramp area} is a connected component of $G_{t, 0}$, which neighbors at least one border cell.

\subsection{Multi-Agent Collective Deconstruction}

The Multi-Agent Collective Deconstruction (MACD) is a planning problem for the deconstruction of the target structure in a blocks-world using a team of agents. The agents have the same set of action types as in the MACC ($\mathcal{A}_\text{MACD} = \mathcal{A}$) and work on the same-sized grid $\mathcal{C}_\text{MACD} = \mathcal{C}$ over the same planning horizon $\mathcal{T}$.

The pick\_up action is the inverse of the deliver action, the inverse to entry action is the leave action and the move action can act as inverse to itself (i.e. the robot can move in the opposite direction). The wait action is its own inverse action. Because each action has an inverse, MACD is the mirror problem to the MACC and therefore has a solution if and only if the MACC has a solution. In case the action durations are specified, let each action in MACD have a duration of its inverse action in MACC.

\subsection{MACC complexity}

Let us now focus on the complexity of the decision problem if a given structure is \textit{buildable} in the MACC problem. Once again, because the MACD problem is a mirror problem to the MACC, the structure is \textit{buildable} if and only if it is deconstructable in the MACD problem.

Let $V(G)$ and $E(G)$ denote the set of vertices/edges of graph $G$, respectively. Let $\text{deg}(v)$ denote the degree of node $v$. Let $x(v_i)$ and $y(v_i)$ be the x-coordinate and y-coordinate of vertex $v_i$, respectively.

Let $G_9$ be a grid graph as defined by \cite{gridNPcomplete}. In summary, \cite{gridNPcomplete} defines $G_9$ as an embedding of a planar, degree $\leq 3$, bipartite graph into a grid, where one node is represented by a 3 by 3 grid cell cluster and $G_9$ does not contain any nodes with $\text{deg}(v) < 2$.

\begin{definition}
    The simple ramp $s_r$ is a non-empty path subgraph of $G$ with nodes $v_0, \dots, v_n$, which starts at a border cell and all remaining nodes are not at the border (i.e. $\forall i \in \widehat{n+1} \setminus \{0\}: v_i \in \mathcal{P}_\boxdot$).
\end{definition}

Let us define high-level action \textit{ramp\_deliver\_block}($v_i$) as an agent \textit{entering} at the simple ramp start with a block, \textit{moving} to $v_{i-1}$, \textit{delivering} a block at $v_i$, \textit{moving} back to $v_0$ and \textit{leaving} the grid.

Let \textit{flat ramp projection} be a projection of a simple ramp to two dimensions, where the vertical axis is the grid z-axis and the horizontal axis is the distance from ramp end $v_n$ when moving on the ramp path. Figure \ref{fig:pathBlockPlacement} shows an example of such a projection.

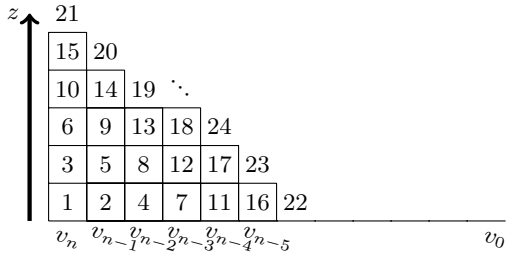
\begin{figure}
\begin{tikzpicture}[scale=.5]

  \begin{scope}
    \draw (0, 0) grid (1, 5);
    \draw (2, 0) grid (1, 4);
    \draw (3, 0) grid (1, 3);
    \draw (4, 0) grid (1, 3);
    \draw (5, 0) grid (1, 2);
    \draw (6, 0) grid (1, 1);
    \draw (12, 0) grid (5, 0);
    \node[anchor=center, rotate=-15] at (0.5, -0.5) {$v_n$};
    \node[anchor=center, rotate=-15] at (1.75, -0.5) {$v_{n-1}$};
    \node[anchor=center, rotate=-15] at (2.75, -0.5) {$v_{n-2}$};
    \node[anchor=center, rotate=-15] at (3.75, -0.5) {$v_{n-3}$};
    \node[anchor=center, rotate=-15] at (4.75, -0.5) {$v_{n-4}$};
    \node[anchor=center, rotate=-15] (1) at (5.75, -0.5) {$v_{n-5}$};
    \node[anchor=center, rotate=-15] (2) at (11.75, -0.5) {$v_0$};
    \foreach \x in {0,...,5} {
        \foreach \y in {0,...,\x}{
            \node[anchor=center] at (0.5+\x-\y, 0.5+\y) {$\thecntr$};
            \stepcounter{cntr}}}
    \node[anchor=center] at (6.5, 0.5) {$\thecntr$};
    \stepcounter{cntr};
    \node[anchor=center] at (5.5, 1.5) {$\thecntr$};
    \stepcounter{cntr};
    \node[anchor=center] at (4.5, 2.5) {$\thecntr$};
    \node[anchor=center] (7) at (4.5, 2.5) {};
    \node[anchor=center] (8) at (2.5, 4.5) {};
    
    \path (7) -- node[auto=false, rotate=-45]{\ldots} (8);
    \draw[->,ultra thick] (-0.5, 0)--(-0.5, 5.5) node[left]{$z$};
  \end{scope}
\end{tikzpicture}
\caption{Order of block placement on a simple ramp.}
\vspace{15pt}
\label{fig:pathBlockPlacement}
\end{figure}

Let us now define the simple ramp-building strategy as starting with the simple ramp $v_0, \dots, v_n$ devoid of blocks and an agent performing \textit{ramp\_deliver\_block}($v_i$) at a position with minimum $L_1$ distance from the $v_n$ bottom block within the \textit{flat ramp projection}. If there is more than one such position, choose the one with the lowest $z$ position. This results in the first placement of the block at $(x(v_n), y(v_n), 0)$ (if $2 \leq n$), with the $L_1$ distance 0. Figure \ref{fig:pathBlockPlacement} shows the following block placement order. Please note, that the strategy can use local block placement to determine the next block position. Let a 2-flat area be an area on the simple ramp with two successive block columns with the same height (even 0). Since the simple ramp strategy places blocks in \enquote{layers} based on their $L_1$ distance, the first 2-flat area when going from $v_n$ to $v_0$ occurs either at the top block of the unfinished next $L_1$ \enquote{layer} or at the simple ramp base (if the top $L_1$ \enquote{layer} is finished, beginning the next one from the bottom). This local placement strategy becomes relevant, when determining block placement position for a block at a ramp with side ramps.

The simple ramp-building strategy allows to place $n$ blocks on the first layer (from $v_1$ to $v_n$), $n-1$ blocks on the second layer (from $v_2$ to $v_n$), and so on, until reaching the $n$-th layer with one block at $v_n$.

Let the length of the simple ramp $l(s_r)$ be the number of edges in its path (i.e. $l(s_r) = n$). Let the ramp maximum height $h_\text{max}(s_r)$ be the number of blocks in the highest column on the ramp. Let the \textit{simple ramp of height $m \in \widehat{n+1}$} be a simple ramp, built using the simple ramp-building strategy, stopped after \textit{ramp\_deliver\_block}($v_i$) places the first block at height $z = m$.

Let us define high-level action \textit{ramp\_pick\_up\_block}($v_i$) as an agent \textit{entering} at the simple ramp start with a block, \textit{moving} to $v_{i-1}$, \textit{picking\_up} the block at $v_i$, \textit{moving} back to $v_0$ and \textit{leaving} the grid. Let the simple ramp-deconstruction strategy be a strategy, that applies \textit{ramp\_pick\_up\_block} to nodes in reverse order relative to the simple ramp-building strategy which would achieve the block placement at the start of deconstruction. The block to be removed can also be determined based on local area, by searching from $v_0$ to $v_{n-1}$ for a node $v_i$ and its successor $v_{i+1}$ with heights $0 < z_{t, x(v_i), y(v_i)} = z_{t, x(v_{i+1}), y(v_{i+1})}$. If no such $v_i$ exists, remove the block at $v_n$. Otherwise, remove the block at $v_i$. Since the block placements match the outcome of a simple ramp-building strategy, the flat area matches the end of the unfinished last $L_1$ layer and $v_i$ the position of the last block placement. If there is no flat block area, the default strategy removes the last block placed on the $L_1$ layer, the top block.

Let us define a high-level action \textit{add\_edge} for a simple ramp $s_r$ and a node $v_{n+1} \in \mathcal{N}_{(x(v_n), y(v_n)} \cap \mathcal{P}_\boxdot$ ($\forall i \in \widehat{n+1}: v_{n+1} \neq v_i$). Let us build a \textit{simple ramp of height $(m-1)$}. We can do that because $m \leq n+1 \Leftrightarrow m-1 \leq n = l(s_r)$. Now add the node $v_{n+1}$ to the end of the simple ramp path. Since the ramp just finished an $L_1$ layer and the new column is one block taller than the end of the ramp, adding the column matches the simple ramp built on path $v_0, \dots, v_{n+1}$. Now deconstruct the simple ramp using the simple ramp deconstruction strategy.

\begin{lemma}
    If there exists a simple ramp $s_r$ of length $n$, with nodes $v_0, \dots, v_n$, with $v_n$ next to a column $(x', y') \in \mathcal{N}_{v_n}$ of the target structure with height $m \leq n+1$, then the column at position $(x', y')$ is deconstructable.
    \label{lbl:deconstructablePath}
\end{lemma}

\begin{proof}
    Perform the \textit{add\_edge} action on the column $(x', y')$. This provides the plan for its deconstruction.
\end{proof}

\begin{lemma}
    The maximum height of any ramp within a \textit{ramp area} of size $n$ is at most $n$ (where size is the vertex count).
    \label{lemma:max_ramp}
\end{lemma}

\begin{proof}
    Robot actions move, pick\_up, and deliver all interact with a neighbor position in graph $G$. A ramp can, therefore, use at most one \textit{ramp area}. Additionally, the deliver action requires two neighbor positions to have the same height, before the new block is placed. This causes each following block layer to be at least one block smaller (the position, where the agent stands, when it places the last block of the layer). The first layer can have up to $n$ blocks (the agent can place the last block from a neighboring border cell). Since each following layer is at least one block smaller, there can be at most $n$ layers with a non-zero number of blocks.
\end{proof}

Let $y_\text{min}$ be the minimum y-coordinate position of any node in $G_9$. Let $V_{y_\text{min}} = \{x : (x, y) \in V(G_9) \wedge y = y_\text{min}\}$ be the set of all vertices of $G_9$ at y position $y_\text{min}$. Let $x_\text{min}$ and $x_\text{max}$ be the minimum/maximum x-coordinate value of any node with a y-coordinate equal to $y_\text{min}$, respectively ($x_\text{min} = \min_{\upsilon \in V_{y_\text{min}}}{x(\upsilon)}$ and $x_\text{max} = \max_{\upsilon \in V_{y_\text{min}}}{x(\upsilon)}$). Let $v$ be the node at $(x_\text{min}, y_\text{min})$, let $v'$ be the node at $(x_\text{max}, y_\text{min})$. Let $w = (x_\text{min} + 1, y_\text{min})$ be the neighbor of $v$. The neighbor at this position must exist because $v$ by definition cannot have any neighbors in the direction of negative x and y axes and $\text{deg}(v) \ge 2$, so $\text{deg}(v) = 2$. Let $u$ be the second neighbor of $v$. Let $w' = (x_\text{max} - 1, y_\text{min})$ be the neighbor of $v'$.

Let us now define the target structure $S_{HC}$ as follows:
\begin{itemize}
    \item Let $(x, y) \in \mathcal{P}_\boxdot, \forall(x, y) \in V(G_9)$ and $y_\text{min} = 1$.
    \item There is a row of border cells at $y = 0$.
    \item $X = x_\text{max} + |V(G_9)| + 1$.
    \item $Y = 2 + \max{\{y : (x, y) \in V(G_9)\}}$.
    \item There is a column of height $|V(G_9)|$ at $v$.
    \item There is a column of height $(|V(G_9)| + 1)$ at $(x_\text{max} + 1, y_\text{min})$ (next to $v'$, where $v'$ did not have a neighbor).
    \item All other columns at positions $v'' \in V(G_9) - \{v\}$ have height 0.
    \item The remaining columns of the structure (except for border cells) have height $(|V(G_9)| + 2)$.
\end{itemize}

\begin{lemma}
    The decision problem of the existence of the Hamilton circuit in the $G_9$ graph is NP-complete.
    \label{lbl:g9NPcomplete}
\end{lemma}

\begin{proof}
    The proof is provided in the original paper describing the $G_9$ graph \cite{gridNPcomplete}.
\end{proof}

\begin{lemma}
    Structure $S_{HC}$ is buildable if and only if there exists a Hamilton circuit in $G_9$.
    \label{lbl:buildableG9}
\end{lemma}

\begin{proof}
    Let us first prove, that if there exists a Hamilton circuit in $G_9$, then $S_{HC}$ is \textit{buildable}.

    The Hamilton circuit must go through all nodes of $G_9$. Because $\text{deg}(v) = 2$, it must go through both of its edges $\{u, v\}, \{v, w\} \in E(G_9)$. If we leave out the edge $\{v, w\}$, we get a Hamilton path of length $|V(G_9)|$, with one end at $w$. and the second end at $v$, with $u$ preceding the node $v$ on this path. Because $w$ is next to a border cell, there exists a path of length $(|V(G_9)|-1)$ from $w$ to $u$. This path leads only over \textit{empty} cells, because the only cell with some blocks is $v$, which is not part of the path. Using lemma \ref{lbl:deconstructablePath}, we deconstruct the column at $v$, leaving \textit{empty} cell.

    Similarly, if we remove the edge $\{w', v'\}$ from the Hamilton circuit, we get a path starting at $w'$ next to a border cell, going over empty cells (including now empty $v$) and ending at $v'$. The length of this path is $|V(G_9)|$. We can now once again use the lemma \ref{lbl:deconstructablePath} to remove the column at $(x_\text{max} + 1, y_\text{min})$ with height $(|V(G_9)|+1)$. Additionally, since by adding a leaf to the end of a path, the graph remains a path, we can use the new path of length $(|V(G_9)|+1)$, including the position $(x_\text{max} + 1, y_\text{min})$ to access position $(x_\text{max} + 2, y_\text{min})$ with height $(|V(G_9)|+1)$. By repeating this process, we can clear the path of length $|V(G_9)|$ along border. This path must exist, because $v'$ is the \textit{empty} cell at maximum x position within the row at $y_\text{min}$ and so all cells in the direction of the positive x coordinate from its position must have height $(|V(G_9)|+2)$ by definition of $S_{HC}$.

    After clearing out the path, we can build a new path $p_\mathcal{B}$ from $(X-2, y_\text{min})$ to $(x_\text{max}, y_\text{min})$ with length $(|V(G_9)|+1)$. Because $p_\mathcal{B}$ is located next to the border, it leaves the rest of the building area a connected graph. We can therefore plan shortest paths from the end of $p_\mathcal{B}$ to the rest of the building area and remove the rest of the columns belonging to the structure $S_{HC}$. Structure $S_{HC}$ is deconstructable in MACD, therefore it is buildable in the MACC.

    Let us now prove, that if the $S_{HC}$ structure is buildable, then there exists a Hamilton circuit in $G_9$.

    Proof by contradiction: Let us assume, that the $S_{HC}$ structure is buildable and there is not a Hamilton circuit in $G_9$. The structure $S_{HC}$ is buildable, if and only if it is deconstructable. There must be a first block that is removed from the structure. This block must be the top block of the column at $v$, because it is the only top of a column reachable according to the maximum ramp height given by lemma \ref{lemma:max_ramp}, which is $(|V(G_9)|-1)$ (the size of the ramp area), so it can reach only columns with height at most $|V(G_9)|$). Let us assume that such a ramp can be built and that it can reach $v$. This ramp cannot be a path, because then we could use the path of the ramp, along with the edge $\{v, w\}$ as a Hamilton circuit. This means, when deconstructing in MACD (and suppose that the agent managed to move next to the top block and pick it up), that the agent cannot carry its block from the top to a border cell without moving the ramp blocks along the way. But the only two positions, where the agent can place its block are the top of the ramp $v_{n+1}$ and its neighbor $v_{n}$, since each layer below has only one block with an unoccupied top (according to the proof of lemma \ref{lemma:max_ramp}). Since the surrounding structure columns have height $(|V(G_9)| + 2)$, the agent cannot climb on top of the structure. And since the agent holds a block, it cannot move any part of the ramp without leaving its block back at the top. So the only way for the agent to carry the block downward is to straighten the ramp into a path, which is in contradiction with the assumption that the ramp is not a path.
\end{proof}

\begin{lemma}
    The decision problem, if the given structure is \textit{buildable} in MACC, is NP-hard.
    \label{lemma:buildableNPhard}
\end{lemma}

\begin{proof}
    Because an NP-complete problem (see lemma \ref{lbl:g9NPcomplete}) can be polynomially transformed into the decision problem if a structure is \textit{buildable} (see lemma \ref{lbl:buildableG9}), the buildability decision problem is NP-hard.
\end{proof}

We must note, that the only reason for NP-hardness, is that the agent cannot climb down with its block, because there are no places below the top for it to lay its block and reconstruct the ramp. If we provide such a flat area, there may be a way to still build to a reasonable height while not requiring the search for the Hamilton path. And as we show in the next section, such a way exists and we use it as a basis for our ReRamp algorithm.


\section{ReRamp algorithm}
The previous section shows that the full utilization of a \textit{ramp area} is possible only when there exists a Hamiltonian path, making the problem of structure constructability NP-hard (see lemma \ref{lemma:buildableNPhard}). However, there is a way if we do not require the full use of the space. We propose a novel construction strategy, utilizing off-path space by constructing reversible ramps to the sides of the main ramp. Each side ramp, when reversed, allows the robot to continue from its top (shown in figures \ref{fig:forward-side-ramp} and \ref{fig:backward-side-ramp}). This strategy permits the agent to build much higher ramps, without the need to search for a Hamilton path.

\begin{figure}
    \centering
    \includegraphics[width=0.65\columnwidth]{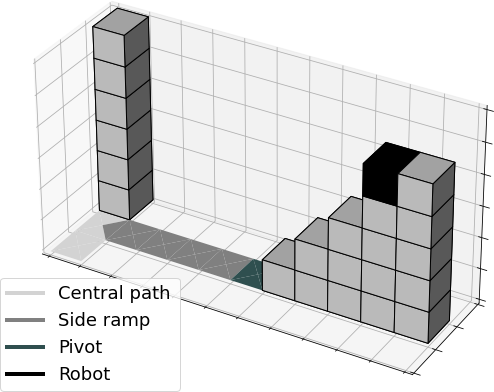}
    \vspace{5pt}
    \caption{Ramp with forward side ramp.}
    \vspace{10pt}
    \label{fig:forward-side-ramp}
\end{figure}

\begin{figure}
    \centering
    \includegraphics[width=0.65\columnwidth]{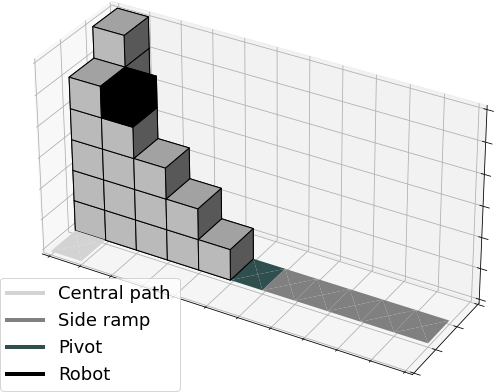}
    \vspace{5pt}
    \caption{Ramp with reversed (backward) side ramp.}
    \vspace{15pt}
    \label{fig:backward-side-ramp}
\end{figure}

\begin{definition}
    The ramp central path (denoted by $p$) is a non-empty sequence of positions $v_0, \dots, v_n$, where $v_0 \in \mathcal{P}$ and $\forall i \in \{1, \dots, n\}: v_i \in \mathcal{P}_\boxdot$. The connection point $\connectionpoint{p} = v_0$ is the first position of $p$. Let $l_p$ be the last position of $p$.
\end{definition}
\begin{definition}
    The (compound) ramp is a tuple $(p, \mathcal{S}, f)$, where $p$ is the ramp's central path, $\mathcal{S}$ is a sequence of reversible ramps connected to the ramp $(p, \mathcal{S}, f)$ by the first node of their central path and $f \in \mathbb{N}$ is the floor height of the ramp. Let $\mathcal{M}$ be the set of all ramps.
\end{definition}
Let the relative height of a ramp $r = (p, \mathcal{S}, f)$ at timestep $t$ and point $(x, y)$ be $\rh{t}{r}{x, y} = z_{t, x, y} - f$. Let the relative reversible height $\rrh{t}{s}$ of reversible ramp $s$ at timestep $t$ be the relative height of $s$ in its backward state at timestep $t$, if all its backward side ramps were also reversed.
\begin{definition}
    The ramp $(p, \mathcal{S}, f)$ is \textit{valid}, iff all its side ramps are reversible and at any given timestep $t \in \mathcal{T}$ all positions $(x, y) \in p$ with a successor $(x', y') \in p$ meet the condition $0 \leq z_{t, x', y'} - z_{t, x, y} \leq 1 + \sum_{s \in \mathcal{S}: \connectionpoint{s} = (x, y)}{}$.
\end{definition}
\begin{definition}
    A ramp $s = (p, \mathcal{S}, f)$ is reversible if it is \textit{valid} and there exists a sequence of actions for a tuple of ramps $((p_f, s_f, f), (p_b, s_b, f))$ to move all the blocks from the first ramp to the second ramp. In the tuple, $(p_f, s_f, f) \in \mathcal{M}$ is the forward ramp (starting with all the blocks belonging to $s$) and $(p_b, s_b, f) \in \mathcal{M}$ is the backward ramp (starting without blocks). Both share the floor height with the ramp $s$.
\end{definition}

To move along the central path of a ramp up, reverse one-by-one all the side ramps connecting to the current agent position (in the order they appear in $\mathcal{S}$) to the backward state, stand at the highest block of the last reversible side ramp and go to the next central path node. To move along the central path down, perform this process in reverse.

Let ramp capacity be the number of blocks on the ramp above its floor height when the ramp is considered to be full. Let reversible capacity be the capacity of the backward ramp $(p_b, s_b, f)$ in the reversible ramp. We consider the ramp full when all side ramps are at reversible capacity, the difference between each node $v_i$ of the central path and its successor $v_{i+1}$ is equal to $(1 + \sum_{s' \in \mathcal{S}: \connectionpoint{s'} = v_i}{\mrrh{s}})$ and, when all the potential side ramps connected to the last node $l_p$ are reversed, the last one reaches $\sum_{s' \in \mathcal{S}: \connectionpoint{s'} = l_p}{\mrrh{s}})$ relative to the floor height of the first one.

Let the maximum reversible height $\mrrh{s}$ of ramp $s$ be the maximum theoretical $\rrh{t}{s}, \forall t \in \mathcal{T}$, which we know is still reversible. In our case, we get this by extending the backward ramp $(p_b, s_b, f)$ from one side of the central path of $s$ and extending the forward ramp from the other side of the central path of $s$, while keeping the backward ramp at most at equal capacity to the forward ramp, until the connection points of both ramps meet at the pivot. We then assign each of the side ramps connected at the pivot to either the backward ramp or the forward ramp to maximize the backward ramp height while keeping the capacity constraint. We then repeat this procedure (the growth of the backward ramp) for the backward ramp without side ramps, the forward ramp still has side ramps. With central path $v_0, \dots, v_n$ of length $n$, the backward ramp has height at least $\lfloor n / 2\rfloor$ blocks (because the worst-case forward ramp is without side ramps, which also has $\lfloor n / 2\rfloor$ block height). Choose the backward ramp with the larger height for the reversible ramp.

Let a 2-flat area on a ramp be two successive nodes, where the height of the first one, together with the relative reversible height of the side ramps connected at the first node $v_i$, is equal to the height of the second node $v_{i+1}$ (without side ramps).

\begin{definition}
    The ramp-building strategy in the first step examines the ramp's central path in the backward direction and searches for a 2-flat area. When found, it marks the second node as a target for block placements.

    When a position on the central path is marked for block placements, newly arrived blocks are placed first on the side ramps connected to the marked position, in the reverse order of $\mathcal{S}$ -- if the ramp is not at reversible capacity, then to expand the ramp, otherwise to raise the side ramp floor (placing a block at every column and increasing $f$ by one). Once all ramps are at their capacity and have raised floor, place the block at the target block placement position on the central path and unmark the position. If an expanded side ramp reaches a new reversible relative height, also unmark the position.

    Return to the first step, until the ramp reaches its capacity (reversible capacity in case of reversible ramps).
\end{definition}

The high-level action \textit{add\_edge} for the ramp consists of four steps. First, we create an empty ramp based on a pre-computed spanning tree of a ramp area. Next, we use the ramp-building strategy to increase the ramp height until it reaches (including its reversed side-ramps) height $(m-1)$ (where $m$ is the height of the column targeted for deconstruction). As the third step, we append the target column to the central path. Finally, we use the reverse version of the ramp-building strategy (i.e. the ramp deconstruction strategy) to remove all blocks from the newly extended ramp.

The first step is new, so it requires a more detailed look. At the start, we have a spanning tree of the \textit{ramp area}, with an additional leaf node in the border cell $v_0$ and the position $v_{n+1}$ of the target column we intend to disassemble. We first create the ramp on the spanning tree using the Ramp from tree (RFT) function (see algorithm \ref{alg:rft}; the connection point is the spanning tree root). In short, we use the path from the border cell to the neighbor $v_n$ of $v_{n+1}$ as the central path and we apply the algorithm \ref{alg:rft} to the longest path for each of the side branches (trivial in a spanning tree) until no side branches remain or the maximum recursion of the algorithm \ref{alg:rft} is reached. The maximum recursion of the algorithm \ref{alg:rft} is given by the user and is important for the complexity analysis of the ReRamp algorithm.

\begin{algorithm}
\caption{Ramp from tree (RFT) function.}
\label{alg:rft}
    \SetAlgoLined
    \KwData{rooted tree graph, $v_n$, max side-ramp recursion $i_r$}
    \KwResult{ramp from root node to $v_n$}
    initialization\;
    $v_i \gets v_n$\;
    \While{$v_i$ is not root}{
        add $v_i$ to central path $p$\;

        \If{$i_r > 0$}{
            \ForEach{child node $v_c$ of $v_i$, not in central path $p$}{
                $\text{child\_tree} \gets \text{rooted tree graph with }v_i\text{ as root}$\;
                $n_d \gets \text{farthest node from }v_c\text{ not crossing }v_i$\;
                $\text{side\_ramp} \gets \text{RFT(child\_tree, }n_d\text{, }i_r-1\text{)}$\;
                \If{side\_ramp reversible height > 0}{
                    add side\_ramp to $\mathcal{S}$\;
                }
            }
        }
        
        $v_i \gets \text{parent(}v_i\text{)}$\;
    }
    add $v_i$ to central path $p$\;
    \Return{ramp $(p, \mathcal{S}, 0)$}\;
    \vspace{5pt}
\end{algorithm}

Now we finally get to the ReRamp algorithm itself. The algorithm \ref{alg:multireramp} consists of a single-agent procedure applied to the decomposition of the target structure to achieve multi-agent parallelism. The decomposition is done by assigning each node to the closest border cell. One agent is then assigned to each border cell and its associated area. For each agent, repeat disassembly of the target structure using the single agent ReRamp subroutine (algorithm \ref{alg:singlereramp}) until no change in structure occurs (the area is \enquote{frozen}). While there is more than one agent, connect two neighbor frozen areas (the ones, which became frozen first), unfreeze the new area, and return to disassembling it until it is frozen again. Since the grid is a connected graph, we end with one large frozen area and one agent. Therefore, the worst-case analysis of deconstructable structure height for the single-agent ReRamp subroutine also applies to the multi-agent ReRamp algorithm.

The single-agent ReRamp subroutine (algorithm \ref{alg:singlereramp}) first constructs a spanning tree using a modified depth-first search algorithm. The modified DFS (mentioned as 
\enquote{simple\_deconstruct\_DFS} in algorithm \ref{alg:singlereramp}) starts at a border cell and encompasses the ramp area neighboring the border cell. It also encompasses the structure columns deconstructable using a simple ramp from the border cell to the target structure column, going over the spanning tree. The simple ramp is constructed up to the neighbor of the target column and the \textit{add\_edge} high-level action is used to include the structure column in the simple ramp (as well as in the spanning tree). However, the simple ramp is not immediately fully deconstructed -- if the following node of the DFS is neighboring the simple ramp end $v_{n+1}$ and is also deconstructable using a simple ramp over the spanning tree, then the ramp height is merely adjusted to perform another \textit{add\_edge}. If the simple ramp is too high, blocks are removed using the simple ramp deconstruction strategy. If the simple ramp is too low, blocks are added using the simple ramp-building strategy. If the following node of the DFS is not neighboring the end of the simple ramp, it is fully deconstructed using the simple deconstruction strategy. The full deconstruction also happens at the end of the simple\_deconstruct\_DFS, leaving the spanning tree as an area without blocks.

If a structure column $w_i$, found by the simple\_deconstruct\_DFS, is too high for deconstruction using the simple ramp over the spanning tree, reference to its neighbor is saved, along with $w_i$ position. Once the creation of the spanning tree is complete, the single-agent ReRamp subroutine tries to disassemble every found structure column from each one of its neighbors within the spanning tree, using the RFT function (algorithm \ref{alg:rft}) on the spanning tree (rooted in the border cell), the neighbor node as $v_n$ and max side-ramp recursion given by the user. The resulting ramp is then used by deconstruct\_DFS, in much the same fashion as the simple ramp in simple\_deconstruct\_DFS. In short, for the disassembly of a structure column at the end of the ramp, the ramp is adjusted in height by block additions/removals to ensure its full height reaches one block below the target structure column height. The target column is added to the ramp using \textit{add\_edge} high-level action, and the target column's position is also added to the spanning tree. If the target for disassembly is not at the ramp end, remove all blocks from the ramp and move the ramp end accordingly. This is done until the deconstruct\_DFS algorithm does not exhaust all the structure columns it can reach with the ramp. Then the ramp is fully disassembled.

\begin{lemma}
    The multi-agent ReRamp algorithm has polynomial complexity $O(n^{4.5} \cdot n^{3i_r})$, where $n$ is the construction area size (vertex count) and $i_r$ is the max side-ramp recursion index, set as part of the user input.
    \label{lemma:polynomialReRamp}
\end{lemma}

\begin{proof}
    Let us first determine the complexity of worst-case ramp full construction and deconstruction.

    \begin{itemize}
        \item ramp with $i_r = 0$
        \begin{itemize}
            \item add\_block $O(n)$ \textit{move} actions + $O(1)$ other actions
            \item remove\_block $O(n)$ \textit{move} actions + $O(1)$ other actions
            \item build\_ramp $O(n^2)$ add\_block actions
            
            (i.e. $O(n^3)$ \textit{move} actions)
            \item deconstruct\_ramp $O(n^2)$ remove\_block actions
            
            (i.e. $O(n^3)$ \textit{move} actions)
            \item pass\_up/pass\_down
            
            (reverse as side ramp when moving up/down, respectively)

            deconstruct\_ramp + build\_ramp (i.e. $O(n^3)$ \textit{move} actions)
        \end{itemize}
        \item ramp with $i_r > 0$
        \begin{itemize}
            \item c\_move (compound move) action
            
            $O(1)$ pass\_up actions for $i_r'=(i_r-1)$ sub ramps + $O(1)$ pass\_down actions for $i_r'$ sub ramps + $O(1)$ move actions
            (if we analyze the recursive calls, we get $O(n^{3i_r})$ \textit{move} actions)
            \item add\_block $O(n)$ c\_move actions + $O(1)$ other actions
            \item remove\_block $O(n)$ c\_move actions + $O(1)$ other actions
            \item build\_ramp $O(n^2)$ add\_block actions
            
            (i.e. $O(n^3)$ c\_move actions)
            \item deconstruct\_ramp $O(n^2)$ remove\_block actions
            
            (i.e. $O(n^3)$ c\_move actions)
            \item pass\_up/pass\_down
            
            (reverse as side ramp when moving up/down, respectively)

            deconstruct\_ramp + build\_ramp (i.e. $O(n^3)$ c\_move actions)
        \end{itemize}
    \end{itemize}
    
    Now let us first determine the complexity of the single-agent ReRamp procedure (algorithm \ref{alg:singlereramp}).

    \begin{itemize}
        \item DFS (for every node may completely build and deconstruct a simple ramp with $i_r = 0$)

        $O(n) \cdot (O(n^3) + O(n^3)) = O(n^4)$
        \item All simple ramp blocks must be removed $O(n^3) \cdot O(n^{3i_r}) = O(n^3 n^{3i_r})$
        \item For every remaining column of the structure (i.e. $O(n)$) do
        \begin{itemize}
            \item RFT (algorithm $\ref{alg:rft}$) inspects every node once $O(n)$
            \item DFS (for every remaining column may completely build and deconstruct a ramp with $i_r$ given by the user -- but then the column is gone, so it is done for once from the perspective of the loop)

            $O(n) + O(1) \cdot (O(n^3) + O(n^3)) \cdot O(n^{3i_r}) = O(n^3 n^{3i_r})$
            \item Then all ramp blocks must be removed $O(n^3) \cdot O(n^{3i_r}) = O(n^3 n^{3i_r})$
        \end{itemize}
    \end{itemize}

    The largest complexity has the DFS algorithm in the loop (along with the following ramp deconstruction) with $O(n^4 n^{3i_r})$. After summing all the complexities, the overall asymptotic complexity of the single-agent ReRamp procedure is $O(n^4 n^{3i_r})$.

    The multi-agent ReRamp algorithm (algorithm \ref{alg:multireramp}) first calls the single agent procedure on $O(\sqrt{n})$ sub-areas of size $O(\sqrt{n})$. Then it connects two areas and runs the single agent procedure on the combined area, until all areas are combined into one. In total, it runs $O(\sqrt{n})$ times and calls the single agent procedure with complexity $O(n^4 n^{3i_r})$. The total complexity of the multi-agent ReRamp algorithm is therefore $O(n^{4.5} \cdot n^{3i_r})$. Because $i_r$ is a constant, set by the user, the algorithm \ref{alg:multireramp} is polynomial.
\end{proof}

\begin{lemma}
    The worst-case ReRamp algorithm maximum construction height is $\Omega(\sqrt{|a|})$ (using the Knuth definition \cite{Knuth1976BigOA}), where $|a|$ is the size (vertex count) of ramp area $a$.
    \label{lemma:worst_case_area}
\end{lemma}

\begin{proof}
    When we decide on the reversible ramp's maximum height, we use a backward ramp with or without side ramps depending on which one is higher. Let $l(s)$ be the length (edge count) of the central path of ramp $s$. Let there be a ramp $r$ from a border cell, going through a ramp area, up to a target position $v_n$. Each following node of the ramp central path is at least one block higher. If the ramp fills the whole ramp area, maximum height is reached, and lemma \ref{lemma:worst_case_area} holds.

    For simplicity, we remove all leaves of the spanning tree, which are not part of the central path of ramp $r$, but which are connected to it. This removes at most 4 nodes for each central ramp node (the number of its neighbors) and positive multiplicative constants do not affect the big O notation. In the case of the ReRamp algorithm, the maximum number of removed nodes per central path node is actually 2, as the first node is a border cell and therefore without side ramps, the last node is next to a column targeted for deconstruction, so it is the starting border cell or has a predecessor (in both cases having at most 2 neighbors for side ramps), and all other central path nodes have two neighbors on the central path and therefore at most two side ramps of length one (the removed leaves).
    
    Let there now be a positive number of sub-trees of the spanning tree, which intersect the ramp $r$ central path only by their root and the root is a leaf (being duplicated if there is more than one side ramp connection point there). If we choose the longest path $l_i$ (edge count) within each subtree for a reversible ramp, connected to the ramp $r$ at the sub-tree root, we gain a sub-ramp with height at least $\lfloor l_i / 2\rfloor$ blocks. All other sub-ramp paths within the given sub-tree must be shorter than $l_i$ and therefore must fit within the $L_1$ distance $l_i$ from the sub-tree root.

    The $L_1$ distance $l_i$ creates an area of size $|a|_s = 4\sum_{j=1}^{l_i}{j} = 2l_i(l_i+1)$ around the sub-tree connection point (not counting the point, since it already belongs to the central path of ramp $r$). Let $|a|_s' = 4(l_i+1)^2 > 2l_i(l_i+1) = |a|_s$ be the ramp area upper bound. We raise the ramp $r$ by at least $\lfloor l_i / 2\rfloor \geq (l_i - 1) / 2, \forall l_i \in \mathbb{N}$ blocks of the side ramp (due to the way the backward ramp is constructed). The side ramp has length $l_i > 1$ (root and at least two side ramp nodes due to previous leaf removal). Let $w_s' = (l_i - 1) / 2$ be a lower bound of the estimation, how much can the sub-tree raise the ramp $r$.

    The equation \ref{eqn:worst_case_height_proof} chooses $C=1/12$ and $n_0 = 2$ ($n_0$ is, in this case, the first value of $l_i$), the inequality proves (by definition) that $w_s' \in \Omega(\sqrt{|a|_s'})$. Since we proved that even $g(n)$ lower bound ($w_s'$) is larger than $f(n)$ upper bound ($|a|_s'$) with $C$, the inequality must hold even for the actual functions $g(n)$ and $f(n)$, proving the lemma holds for side-ramp trees. Since the spanning tree covers the whole ramp area (due to the properties of the DFS algorithm used), all the area nodes are used either by the central path of ramp $r$ or side ramp trees. The lemma \ref{lemma:worst_case_area} thus holds for all ramp area nodes.
\end{proof}

\begin{lemma}
    The worst-case ReRamp algorithm maximum construction height is $\Omega(\sqrt{n_a})$ (using the Knuth definition \cite{Knuth1976BigOA}), where $n_a$ is the maximum construction height in the ramp area $a$.
    \label{lemma:worst_case_height}
\end{lemma}

\begin{proof}
    Let $w_a$ be the worst-case ReRamp algorithm construction height in the ramp area $a$. Since the ramp can rise by at most one block for each ramp area node (lemma \ref{lemma:max_ramp}), then $(n_a \leq |a| \wedge w_a \in \Omega(\sqrt{|a|})) \Rightarrow w_a \in \Omega(\sqrt{n_a})$. In other words: Since the lemma \ref{lemma:max_ramp} proved rise by at least a square root of the ramp area, which is at least as large as the maximum height, then the less strict lemma \ref{lemma:worst_case_height} also holds.
\end{proof}

\begin{algorithm}
\caption{Single-agent ReRamp algorithm.}
\label{alg:singlereramp}
    \SetAlgoLined
    \KwData{structure heigh-map, accessible area $\mathcal{a}$, access point $v_0$, max side-ramp recursion $i_r$}
    \KwResult{single-agent plan for constructing the structure}
    initialization\;
    spanning\_tree $\gets$ simple\_deconstruct\_DFS($v_0$, $\mathcal{a}$)\;
    \ForEach{structure column $v_{n+1}'$ within $\mathcal{a}$, neighbored by a spanning tree node $v_n'$}{
        use RFT to create a ramp $r$ from $v_0$ to $v_n'$ over the spanning tree with max side-ramp recursion $i_r$\;
        \If{possible to \textit{add\_edge} from $v_n$ to $v_{n+1}$}{
            run deconstruct\_DFS($v_n$, $\mathcal{a}$, ramp $r$) to expand the spanning tree\;
        }
    }
\end{algorithm}

\begin{algorithm}
\caption{Multi-agent ReRamp algorithm.}
\label{alg:multireramp}
    \SetAlgoLined
    \KwData{structure heigh-map, max side-ramp recursion $i_r$}
    \KwResult{multi-agent plan for constructing the structure}
    initialization\;
    divide the structure into areas by closest access point\;
    \ForEach{area and its access point}{
        \While{structure in area changed}{
            try to deconstruct the structure within the area using the single-agent ReRamp algorithm\;
        }
    }
    \While{more than one area remains}{
        connect two neighboring earliest finished areas into one\;
        \While{structure in new area changed}{
            try to deconstruct the structure within the new area using the single-agent ReRamp algorithm with max side-ramp recursion $i_r$\;
        }
    }
    pad plan for each agent to maximum plan length\;
    reverse all plans;
    \vspace{5pt}
\end{algorithm}


\section{Experiments}
We perform the experiments on a 3 GHz Intel Skylake processor with 16 physical cores and 132 GB of RAM. Our implementation of the ReRamp model is single-threaded and written in Python 3.12. We run the ReRamp algorithm 20 times for each instance to get the mean runtime.

In the first experiment, we focus on the comparison with the state-of-the-art algorithms. We compare our algorithm to the exact solvers of the MACC problem. Since they use generic solvers and do not define the ramp-building strategy, they can -- theoretically -- build to the maximum height for each structure. Practically, they are quite limited by the exponential complexity of their problem description. This is also the reason, why all six of the benchmark instances are small.

In particular, we test our ReRamp algorithm against the MILP and CP models by \cite{koenigexact}, the Fraction time MILP model by \cite{icaart24} (which was tested only against 2 unmodified benchmark instances), 3D decomposition with MILP-based Approach by \cite{decomposition} (denoted \enquote{Decomposition A} in the table \ref{tab:result1}) and Parallel Construction of substructures using 3D decomposition with MILP-based Approach (\enquote{Decomposition B} in the table \ref{tab:result1}). All performance data of the other algorithms is sourced from their respective papers, and the hardware used for the tests is roughly similar (Intel Xeon E5-2660 v3
CPU 2.60 GHz, Intel Skylake 3 GHz, Intel Core i7-7700K CPU 4.20GHz, respectively). All structures are computed with max side-ramp recursion set to 1 (as this is the lowest setting allowing the structure height guarantees, given by lemma \ref{lemma:worst_case_area}).

While the Tree-Based algorithm is in some aspects similar to the ReRamp algorithm -- both are heuristic algorithms, and both work with spanning trees of the building area -- our algorithm is intended to work on higher structures than the Tree-Based algorithm can create. The wider area around the structure, which the Tree-Based algorithm requires, allows the use of simpler ramps and easier parallelization. This wide area may not be available when building real-world structures, especially in urbanized areas. So the modification of at least the workspace matrix of the Tree-Based algorithm would be necessary, which is outside the scope of this paper.

We can see in table \ref{tab:result1}, that our ReRamp algorithm has a much smaller computation time, than all state-of-the-art algorithms able to build its class of structures. This is mainly caused by the fact, that algorithms with comparable capabilities to the ReRamp algorithm in terms of structure height are computed using general solvers with exponential time complexity. In contrast, the ReRamp algorithm is polynomial (see lemma \ref{lemma:polynomialReRamp}) and uses the more time-consuming reversible ramps only in places, where the algorithm considers it necessary. Furthermore, our algorithm displays comparable performance to the Decomposition A algorithm, while being at least an order of magnitude faster to compute.

\begin{table}[htb]
\centering
\caption{Experimental results; comparison of the performance of our ReRamp model with data for both \enquote{Constant time} models by \cite{koenigexact}, data for \enquote{Fraction time} model by \cite{icaart24}, data for both \enquote{Decomposition} models by \cite{decomposition}; visualizations on the left are recreations of the figures used in the papers \cite{koenigexact,icaart24,decomposition}; best value in \textcolor{olive}{green}.}
\resizebox{\columnwidth}{!}{
\begin{tabular}{llllllll}
\hline
Instance & Model & Computation & Makespan & Sum of & Robots \\
 &  & time mean & & costs & \\ \hline \hline
\multirow{6}{*}{1
\begin{minipage}[b]{0.08\textwidth}
    \includegraphics[width=\textwidth]{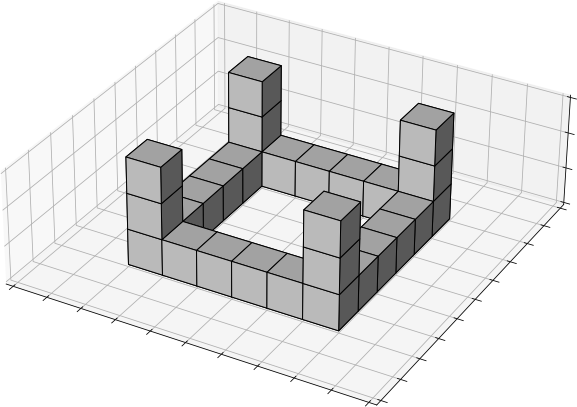}
\end{minipage}
} & Const. time MILP & 29s & \textcolor{olive}{11} & 176 & 34 \\
 & Const. time CP & >7d & \textcolor{olive}{11} & 178 & 30 \\
 & Fraction time & 26.75s & \textcolor{olive}{11} & 232 & 37 \\
 & Decomposition A & 241.3s & 48 & 176 & -- \\
 & Decomposition B & 259.4s & 17 & 179 & -- \\
 & ReRamp & \textcolor{olive}{0.526s} & 47 & 268 & 20 \\ \hline
\multirow{6}{*}{2
\begin{minipage}[b]{0.08\textwidth}
    \includegraphics[width=\textwidth]{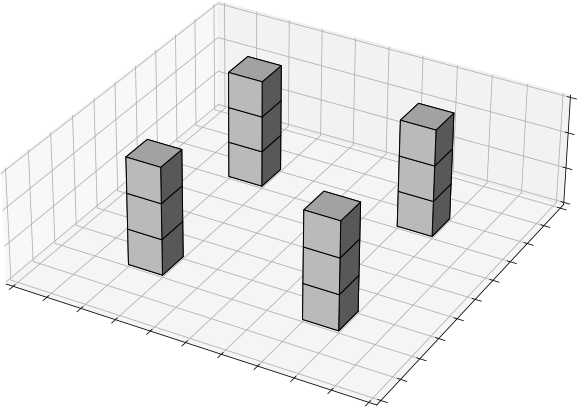}
\end{minipage}
} & Const. time MILP & 3s & \textcolor{olive}{11} & 128 & 28 \\
 & Const. time CP & 1.2h & \textcolor{olive}{11} & 128 & 28 \\
 & Fraction time & 24.10s & \textcolor{olive}{11} & 196 & 32 \\
 & Decomposition A & 235.9s & 48 & 128 & -- \\
 & Decomposition B & 198.1s & 14 & 128 & -- \\
 & ReRamp & \textcolor{olive}{0.465s} & 47 & 188 & 4 \\ \hline
\multirow{5}{*}{3
\begin{minipage}[b]{0.08\textwidth}
    \includegraphics[width=\textwidth, trim={0 0pt 0 0pt}, clip]{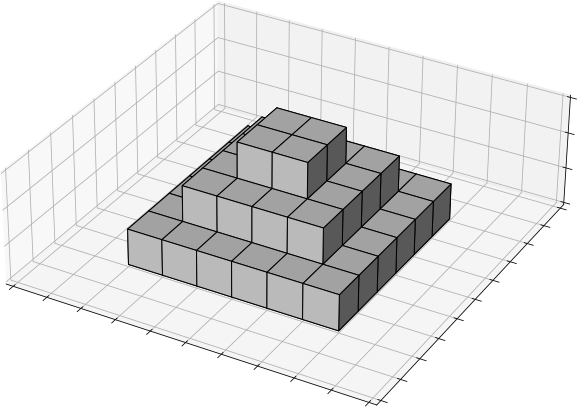}
\end{minipage}
} & Const. time MILP & 1.2h & \textcolor{olive}{13} & 344 & 44 \\
 & Const. time CP & >7d & \textcolor{olive}{13} & 354 & 44 \\
 & Decomposition A & 377.2s & 106 & 326 & -- \\
 & Decomposition B & 318.1s & 44 & 326 & -- \\
 & ReRamp & \textcolor{olive}{0.424s} & 46 & 376 & 20 \\ \hline
\multirow{5}{*}{4
\begin{minipage}[b]{0.08\textwidth}
    \includegraphics[width=\textwidth, trim={0 0pt 0 0pt}, clip]{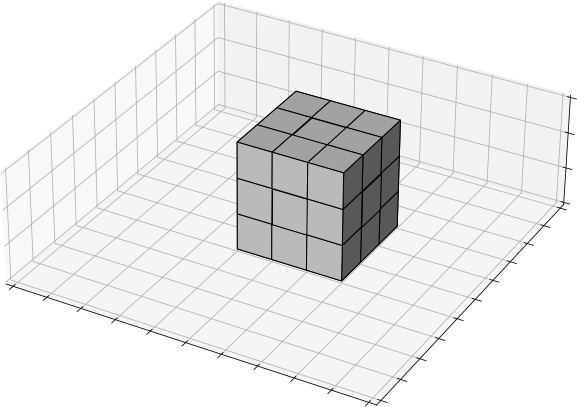}
\end{minipage}
} & Const. time MILP & 5.5h & \textcolor{olive}{17} & 429 & 42 \\
 & Const. time CP & >7d & \textcolor{olive}{17} & 452 & 50 \\
 & Decomposition A & 31.2s & 113 & 204 & -- \\
 & Decomposition B & 758.6s & 75 & 263 & -- \\
 & ReRamp & \textcolor{olive}{0.401s} & 74 & 419 & 7 \\ \hline
\multirow{5}{*}{5
\begin{minipage}[b]{0.08\textwidth}
    \includegraphics[width=\textwidth, trim={0 0pt 0 0pt}, clip]{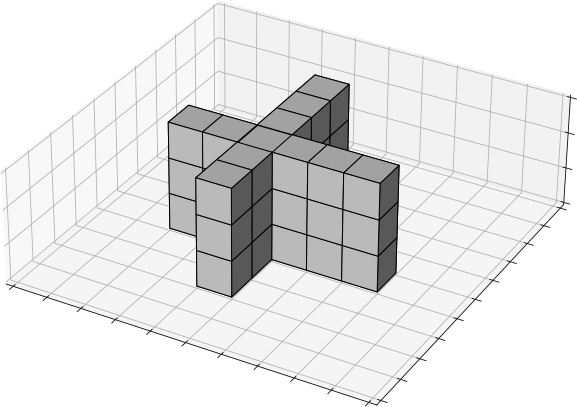}
\end{minipage}
} & Const. time MILP & 5.7d & \textcolor{olive}{17} & 368 & 37 \\
 & Const. time CP & >7d & \textcolor{olive}{17} & 395 & 41 \\
 & Decomposition A & 27.8s & 130 & 365 & -- \\
 & Decomposition B & 688.5s & 90 & 381 & -- \\
 & ReRamp & \textcolor{olive}{0.462s} & 388 & 527 & 3 \\ \hline
\multirow{5}{*}{6
\begin{minipage}[b]{0.08\textwidth}
    \includegraphics[width=\textwidth, trim={0 0pt 0 0pt}, clip]{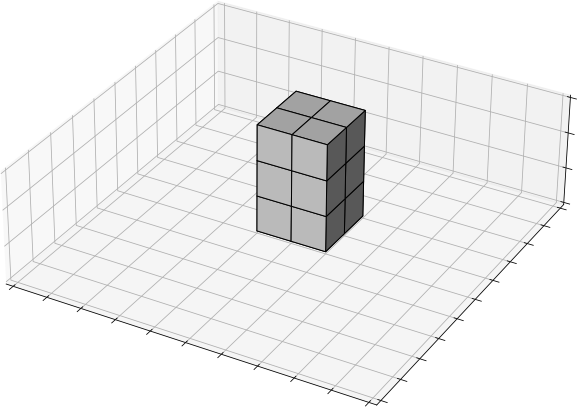}
\end{minipage}
} & Const. time MILP & 183s & \textcolor{olive}{15} & 234 & 27 \\
 & Const. time CP & >7d & \textcolor{olive}{15} & 245 & 28 \\
 & Decomposition A & 12.2s & 50 & 153 & -- \\
 & Decomposition B & 287.9s & 40 & 161 & -- \\
 & ReRamp & \textcolor{olive}{0.417s} & 65 & 224 & 4 \\ \hline
\end{tabular}
}
\label{tab:result1}
\end{table}

In the second experiment, we evaluate the performance of the algorithm on a manually discretized structure of a single-story house (its floor plan is displayed in figure \ref{fig:house}). The discretization is based on the TERMES block size \cite{petersen2011termes} -- 21.5 x 21.5 x 4.5 cm -- with appropriate indents left for windows, doors, and lintels. Each wall is 58 blocks tall (which, if computed by the Tree-Based algorithm \cite{TACR}, would require more than 12m of free space on each side of the house; we use 1.29m for the ReRamp algorithm).  The max side-ramp recursion is again set to 1. The experiment aims to demonstrate the capability of the ReRamp algorithm to create plans for real-life-sized structures.

The ReRamp algorithm completed the plan for the construction of the discretized house in 195.6 minutes (roughly 3 hours). That is less time than was necessary to compute the 3x3x3 cube structure (building area 10x10x3 blocks) in the MILP model by \cite{koenigexact}. In comparison, our house structure is 76x44x58 blocks with an additional border cell and 5 empty blocks from each side, making the full building area 88x56x58 blocks. To build the house, the ReRamp algorithm provides a plan for 24 agents with a makespan of 208 388 442 timesteps and the average sum-of-costs of roughly $474 \cdot 10^6$ timesteps.

The successful creation of this plan proves that the ReRamp algorithm is capable of handling structures with the rough dimensions of a house. This is an important step towards applying multi-agent construction in real-life projects.

\begin{figure}
    \centering
    \includegraphics[width=\columnwidth]{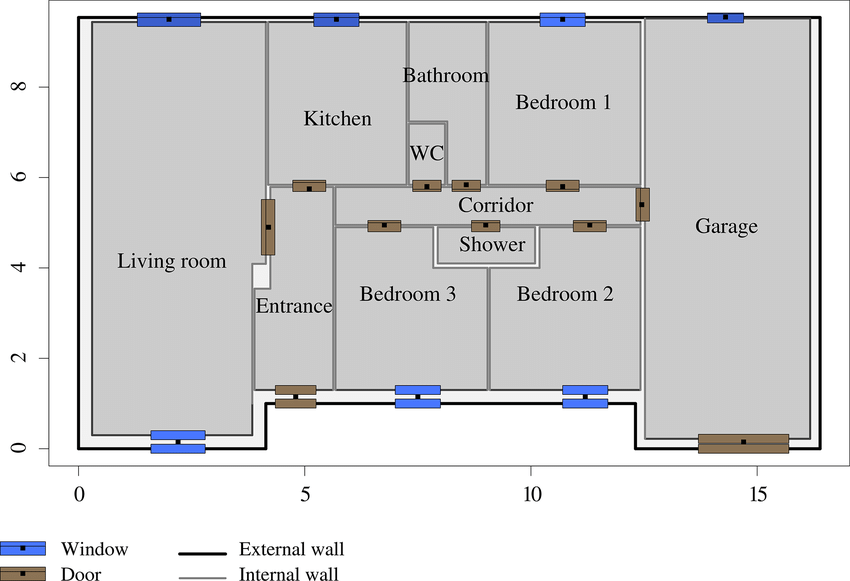}
    \vspace{5pt}
    \caption{\enquote{Top view of the single-storey house.} by Claire Richert, Hélène Boisgontier, and Frédéric Grelot is licensed under CC BY 4.0 (figure 1 of paper \cite{house}).}
    \label{fig:house}
    \vspace{17pt}
\end{figure}


\section{Conclusion}

We propose the ReRamp algorithm, a new algorithm that uses our new concept of reversible ramps to solve problems in multi-agent collective construction. We show, that ramps with a tree footprint are a feasible alternative to path-footprint ramps, allowing a more efficient use of the ramp area at polynomial complexity, which leads to higher structures.

We prove, that our polynomial time algorithm can build to at least $\Omega(\sqrt{n})$ height, where $n$ is the maximum height within a given ramp area and also the size (node count) of the area on the grid. This is a very promising result. It proves that while the decision problem of structure constructability is NP-hard for structures reaching maximum height, we can compute plans for structures with a roughly cubic shape (width, depth, and height) in polynomial time (since the area is $n = \text{width} \cdot \text{depth}$). We also demonstrate this by experimentally generating a plan for a discretized 1:1 structure of a single-story house.

We use a set of benchmark structures to compare our algorithm to state-of-the-art algorithms, which can build to the same height as our algorithm. The experiments show at least an order of magnitude improvement in computational speed, even on small structures. This is expected, as all the current state-of-the-art algorithms, capable of building to such heights, are reliant on generic solvers, making them exponential-time relative to structure size. While the solutions of our algorithm are not optimal, on the benchmark instances they stay mostly within one order of magnitude from the optimal solution. We consider this an acceptable tradeoff for having a new polynomial algorithm, capable of building to such height.

In future work, we propose to focus on further distribution of work among agents. The optimal planners used during experiments may work as a lead, which parts of construction are not yet spread enough among the agents. The TERMES system may also have to be modified to handle different building materials and to ensure block column stability, as the current foam blocks may not be well suited for large structures.



\begin{ack}
This research was supported by GA\v{C}R - the Czech Science Foundation under the project number 22-31346S.
\end{ack}



\bibliography{paper_arxiv}

\end{document}